\documentclass{article}
\usepackage{a4wide}
\usepackage{graphicx}
\usepackage{epsfig}
\usepackage{color}
\usepackage{transparent}
\usepackage{amsmath}
\usepackage{amssymb}
\usepackage{enumerate}
\usepackage{float}
\usepackage{color}
\usepackage{transparent}

\newcommand{\lungh}{w}

\newcommand{\dist}{\mathrm{dist}}

\newcommand{\mincut}{\mathrm{minCut}}
\newcommand{\maxflow}{\mathrm{maxFlow}}
\newcommand{\flowvit}{\mathrm{flowVit}}
\newcommand{\RRe}{\rm I\!R}

\newcommand{\edge}[2]{(#1,#2)}

\newcommand{\commento}[1] {}
\newcommand{\comment}[1] {}

\newtheorem{theorem}{Theorem}

\newtheorem{lemma}[theorem]{Lemma}
\newtheorem{corollary}[theorem]{Corollary}

\newtheorem{proposition}{Proposition}

\newtheorem{definition}{Definition}
\newcommand{\qed}{\hfill \ensuremath{\Box}}

\newenvironment{proof}{\vspace{1ex}\noindent{\bf Proof.}\hspace{0.5em}}
	{\hfill\qed\vspace{2ex}}

\newcounter{progcount}
\newcounter{linecount}[progcount]

{
    \end{tabbing}
    \end{minipage}\\[0.5ex]
}

\ignorespaces
\begin{document}

\title{Max flow vitality in general and $st$-planar graphs}

\author{Giorgio Ausiello\footnote{Dipartimento di Ingegneria Informatica, Automatica e Gestionale, Universit\`a
di Roma ``La Sapienza'', via Ariosto 25, 00185 Roma, Italy. Email: \texttt{ausiello@diag.uniroma1.it}.}
\and 
Paolo G. Franciosa\footnote{Dipartimento di Scienze Statistiche, Universit\`a di Roma ``La Sapienza'',
piazzale Aldo Moro 5, 00185 Roma, Italy. Email: \texttt{paolo.franciosa@uniroma1.it, isabella.lari@uniroma1.it}.} 
\and 
Isabella Lari\footnotemark[2]
\and 
Andrea Ribichini\footnote{Dipartimento di Fisica, Universit\`a
di Roma ``La Sapienza'', piazzale Aldo Moro 5, 00185 Roma, Italy. E-mail: \texttt{ribichini@diag.uniroma1.it}.}
}

\date{}

\maketitle

\begin{abstract}
The \emph{vitality} of an arc/node of a graph with respect to the maximum flow between two fixed nodes $s$ and $t$ is defined as the reduction of the maximum flow caused by the removal of that arc/node.
In this paper we address the issue of determining the vitality of arcs and/or nodes for the maximum flow problem. We show how to compute the vitality of all arcs in a general undirected graph by solving only $2(n-1)$ max flow instances and, In $st$-planar graphs (directed or undirected) we show how to compute the vitality of all arcs and all nodes in $O(n)$ worst-case time. Moreover, after determining the vitality of arcs and/or nodes, and given a planar embedding of the graph, we can determine the vitality of a ``contiguous'' set of arcs/nodes in time proportional to the size of the set.
\end{abstract}
\vspace{1cm}

\section{Introduction}\label{se:intro}
Given a graph with capacities associated to arcs, and given two special nodes $s$ and $t$, the problem of determining the maximum flow (max-flow) that can be transferred from $s$ to $t$ has been deeply studied since the 1950's. Here we present algorithms for computing how the maximum flow is influenced by the removal of any single arc, or any single node, or in some cases by the simultaneous removal of a set of arcs/nodes.
This is a special case of the vitality concept:
given a real-valued function $f(G)$ of a graph $G$,
the \emph{vitality} of a resource $x$ of the graph is usually defined as the value $|f(G) - f(G\setminus x)|$, where $G \setminus x$ denotes the graph after the removal of resource $x$.
Vitality can be seen as a \emph{centrality index}, as defined in~\cite{dagstuhl}.

The vitality of arcs and nodes in a graph has been studied with respect to the distance between two fixed nodes $x,y$~\cite{Bar-noy,CorleySha,Lin,Malik,McMasters,NardelliEdge,NardelliNode,Stahlberg}.
In this case, the vitality of an arc measures the distance increase between $x$ and $y$ when the arc is removed, and obviously only arcs on a shortest path from $x$ to $y$ may have vitality greater than zero.
In~\cite{Ausiello}, the problem of determining a spanner of a graph with the additional constraint of preserving the vitality of arcs with respect to distances has been addressed.

Vitality of arcs with respect to max-flow has been studied since 1963, only a few years after the seminal paper by Ford and Fulkerson~\cite{Fulkerson} in 1956.
Wollmer~\cite{Wollmer63} presented a method for determining the most vital link (i.e., the arc with maximum vitality) in a railway network.
A more general problem has been studied in~\cite{Ratliff}, where an algorithm based on an enumerative approach is proposed for finding the $k$ arcs whose simultaneous removal causes the largest decrease in max-flow. Wood~\cite{Wood} has shown that this problem is NP-hard in the strong sense, while its approximability has been studied in~\cite{Altner,PhillipsSTOC93}.
Corley and Chang~\cite{Corley} have shown that removing nodes can be reduced to removing arcs in a transformed network. A recent survey on vitality of arcs and related problems is in~\cite{Alderson}.

A slightly different problem consists in determining a ``robust'' flow assignment, i.e., a flow assignment in which the flow loss due to the removal of one arc is minimized.
In this setting, flow is not ``re-routed'' due to the arc removal.
Aneja, Chandrasekaran and Nair~\cite{Aneja} propose a strongly polynomial solution based on Linear Programming, while the extension of this problem to the removal of $k$ arcs has been shown to be NP-hard in~\cite{Du}, even for $k=2$.
Recently, the most vital arc or set of $k$ arcs in a flow network which carries flow over time has been studied in~\cite{MOROWATI}.

\comment{\subsection{Our results}
}
Despite the abundant literature on most vital arcs, the problem of efficiently determining the vitality of all arcs or all nodes has not been addressed yet.

In this paper we tackle this problem for general undirected graphs and for $st$-planar graphs. In particular, we obtain the following results ($n$ is the number of nodes of the graph, $m$ is the number of arcs):

\begin{description}
\item[general undirected:] vitality of all arcs in $O(n \cdot \mbox{MF}(n,m))$ time, where $\mbox{MF}(n,m)$ is the time needed to solve a max-flow instance on an undirected graph with $n$ nodes and $m$ arcs;

\item[$st$-planar, both undirected and directed:] vitality of all nodes and all arcs in $O(n)$ time. Moreover, fixing an $st$-planar embedding and after $O(n)$ preprocessing time using $O(n)$ space, we retrieve the vitality of any contiguous
set of $k$ arcs in $O(k)$ time.
\end{description}

Observe that even simple problems concerning maximum flow vitality are at least as hard as the maximum flow problem itself. In fact, the following relations hold:
\begin{proposition}\label{rem:vitality}
Computing the vitality of a single arc is equivalent, in the worst case, to computing the value of the maximum flow.
\end{proposition}
\begin{proof}
The vitality of $e$ is computed by definition as the maximum flow in $G$ minus the maximum flow in $G\backslash e$. On the other hand, given a  graph $G$ with terminal nodes $s$ and $t$, build a graph $G'$ by adding a new node $t'$ and an arc $e=(t, t')$ with large capacity (e.g., larger than the sum of all capacities in $G$). The maximum flow from $s$ to $t'$ equals the flow through $e$. If we compute the vitality of arc $e$ w.r.t.\ the maximum flow from $s$ to $t'$, then we also obtain the maximum flow value from $s$ to $t$.
\end{proof}
\begin{proposition}
Deciding whether a fixed arc is a most vital arc is at least as difficult as deciding whether the value of the maximum flow is greater than a given value.
\end{proposition}
\begin{proof}
Consider the graph $G'$ defined as in Proposition~\ref{rem:vitality} above, and add to $G'$ an arc $g=(s, t')$ with capacity $F$. Arc $e$ is a most vital arc w.r.t.\ the maximum flow from $s$ to $t'$ if and only if the maximum flow from $s$ to $t$ in $G$ exceeds $F$, otherwise the most vital arc is $g$.
\end{proof}

\noindent The above propositions show that our solution in $st$-planar cases is optimal. 
\medskip

The paper is organised as follows: in Section~\ref{se:defs} we provide some definitions and preliminary considerations, the general undirected case is dealt with in Section~\ref{se:undirected}, while Section~\ref{se:stplanar} shows how the problem can be solved in the case of directed or undirected $st$-planar graphs. Final considerations and open problems are given in Section~\ref{se:conclusions}.

\section{Definitions and preliminaries}\label{se:defs}
We are given a weighted directed graph $G=(N,A,c)$, where $N$ is a set of $n$ nodes, $A \subseteq N \times N$ is a set of $m$ arcs, and $c: A \rightarrow \RRe^+$ is a non negative function that assigns \emph{capacities} to arcs. We fix two special nodes $s$ and $t$, and we assume $G$ is connected---i.e., the underlying undirected graph is connected.
A \emph{feasible flow assignment} from $s$ to $t$ is a function $f: A \rightarrow \RRe^+$ such that:
\begin{itemize}
\item $0 \leq f(e) \leq c(e)$ for each $e \in A$ (capacity constraint),
\item $\sum_{x \in N^{-}(v)}f((x,v)) = \sum_{y \in N^{+}(v)}f((v,y))$, for each $v \in N \setminus \{s,t\}$ (conservation constraint),
\end{itemize}
where $N^{-}(v)$ (resp., $N^{+}(v)$) is the set of nodes $\{x \in N)\ | (x,v) \in A\}$ (resp., $\{x \in N\ | (v,x) \in A\}$).

The capacity constraint ensures that the flow on each arc does not exceed its capacity, while the conservation constraint ensures that for each node, other than $s$ or $t$, the flow entering the node equals the flow leaving the node.
The \emph{flow from $s$ to $t$} under a feasible flow assignment $f$ is defined as
$F(f) = \sum_{y \in N^{+}(s)}f((s,y))  -  \sum_{x \in N^{-}(s)}f((x,s))$.
The $\emph{maximum flow}$ from $s$ to $t$ is the maximum value of $F(f)$ over all feasible flow assignments $f$.
A flow assignment giving a maximum flow is called a \emph{maximum flow assignment}. Since $s$ and $t$ are usually fixed, in the sequel we do not specify ``from $s$ to $t$'', and we denote the maximum flow on $G$ from $s$ to $t$ simply by $\maxflow(G)$.

In an undirected graph $G=(N,A,c)$ (we adopt the same notation for directed and undirected graphs, as in~\cite{Ahuja}), a feasible flow assignment $f_{u}$ can be defined by considering a feasible flow assignment $f$ on the directed graph obtained by substituting each undirected arc $\edge{x}{y} \in A$ with a pair of directed arcs $(x, y)$ and $(y,x)$, both having the same capacity.
The flow value $f_{u}(x,y)$ from $x$ to $y$ on the undirected arc $\edge{x}{y}$ is defined as $f_{u}(x, y) = f((x,y)) - f((y,x))$. Note that, in undirected graphs, the flow on each arc is ``directed'', so that $f_{u}(x, y) = - f_{u}(y, x)$.

An \emph{$st$-cut} of a graph is a partition of nodes into two subsets $S,T$ such that $s \in S$ and $t \in T$. In a directed graph, a cut is also identified by the set of arcs from $S$ to $T$, i.e., the set  $A  \cap (S \times T)$. Given a cut $C = (S,T)$ and an arc $e \in A  \cap (S \times T)$, we say that $e$ \emph{crosses} $C$, and $C$ \emph{crosses} $e$, as well. The capacity of a cut is defined as $c(S,T) = \sum_{e \in (A  \cap (S \times T))} c(e)$, and a \emph{minimum $st$-cut} is an $st$-cut having minimum capacity.
Note that the capacity of cuts is oriented, so that in general $c(A,B) \not= c(B,A)$.
In the undirected case, the cut is also identified by the set of arcs having one endpoint in $S$ and the other endpoint in $T$.
Since nodes $s$ and $t$ are fixed, we denote a minimum $st$-cut in graph $G$ simply as $\mincut(G)$.
The well known Min-Cut Max-Flow theorem \cite{Fulkerson} states that $\maxflow(G) = c(\mincut(G))$, for any weighted graph $G$.
Given an arc $e=(x,y)$, by $\mincut_e(G)$ we denote a minimum capacity cut among all $st$-cuts of $G$ that cross $(x,y)$. Obviously, $\mincut_{e}(G)$ is not necessarily a minimum cut, so  $c(\mincut_{e}(G)) \geq c(\mincut(G))$, and $c(\mincut_{e}(G)) = c(\mincut(G))$ if and only if $e$ belongs to some minimum cut.

In what follows, given an arc $e \in A$, we simply denote the graph $G = (N,A\setminus\{e\},c)$ by $G - e$
and, given a node $v \in N$, we denote by $G-v$ the subgraph induced by $N \setminus \{v\}$.
The same notations are extended to sets of arcs or nodes, thus $G - A'$, with $A' \subseteq A$, is the graph $(N, A \setminus A', c)$ and $G - N'$, with $N' \subseteq N$, is the subgraph induced by $N \setminus N'$.
The distance $\dist_{G}(x,y)$ from node $x$ to node $y$ is the length of a shortest path in $G$ from $x$ to $y$. We extend the definition of distance to pairs of node sets, so that $\dist_{G}(A, B) = \min_{u \in A, v \in B}\dist_{G}(u,v)$ for any sets $A,B$ of nodes. In particular, the distance from a node $v$ to an arc $e=(x,y)$ is $\dist_{G}(v, e) = \dist_{G}(\{v\}, \{x,y\})$.

The \emph{vitality} of a resource $R$ with respect to maximum flow, according to the general concept of vitality in \cite{dagstuhl}, is defined as $\flowvit(R) = \maxflow(G) - \maxflow(G-R)$, where $R$ can be a single arc/node, or a set of arcs/nodes. 

Given a planar embedded directed graph $G$, its \emph{dual} graph $G^*$ is defined as a directed weighted multigraph, possibly having self-loops and parallel arcs, whose nodes correspond to faces of $G$ and such that for each arc $e = (x,y)$ in $G$ there is an arc $e^{*} = (f^{*},g^{*})$   in $G^{*}$, where $f^{*}$ corresponds to the face $f$ to the left of $e$ in $G$ and $g^{*}$ corresponds to the face $g$ to the right of $e$ in $G$. The length $\lungh(e^{*})$ of $e^{*}$ equals the capacity of $e$. For each arc $e$ we also include in $G^{*}$ a reverse arc of $e^{*}$, i.e., arc $(g^{*},f^{*})$, whose length is set to 0. We also say that $G$ is the \emph{primal} graph of $G^{*}$. The dual graph of a planar embedded undirected graph is a planar undirected multigraph, and is defined analogously (reverse arcs are not needed).
It is well known that $G^{*}$ is a planar graph, and that duality also maps each node $v$ in $G$ to a face $v^{*}$ in $G^{*}$, and each face $f$ in $G$ to a node $f^{*}$ in $G^{*}$.
Remember that, thanks to Euler's formula, in the case of planar graphs $m=O(n)$.

%
An $st$-planar graph is a planar graph that admits an $st$-planar embedding, i.e., a planar embedding with nodes $s$ and $t$  lying
on the same face. W.l.o.g., we assume $s$ and $t$ are on the outer face. An $st$-planar embedding of an $st$-planar graph can be found in $O(n)$ worst-case time, by computing a planar embedding as in~\cite{Hopcroft74} of the graph after adding arc $(s,t)$.
Given an $st$-planar embedded graph, in which $s$ is placed to the left and $t$ is placed to the right (see the left part of Figure~\ref{fi:dual}), we draw two semi-infinite lines from $s$ to the left and from $t$ to the right, splitting the outer face into an upper face $U$ and a lower face $L$, thus $U^{*}$ and $L^{*}$ will be two special nodes in $G^{*}$.

For any arc $e$ in a planar graph $G$, we denote by $G^{*}_{e}$ the graph obtained from $G^{*}$ by setting $\lungh(e^*)$ to zero. This is equivalent, with respect to path lengths, to contracting arc $e^{*}$ (see the right part of Figure~\ref{fi:dual}). Since arc lengths are non-negative, it follows that $\dist_{G^{*}_{e}}(x,y) \leq \dist_{G^{*}}(x,y)$, for any arc $e$ and any pair of nodes $x,y$ in $G^{*}$. 

\begin{figure}[t]
\begin{center}
\def\svgwidth{14.5cm}
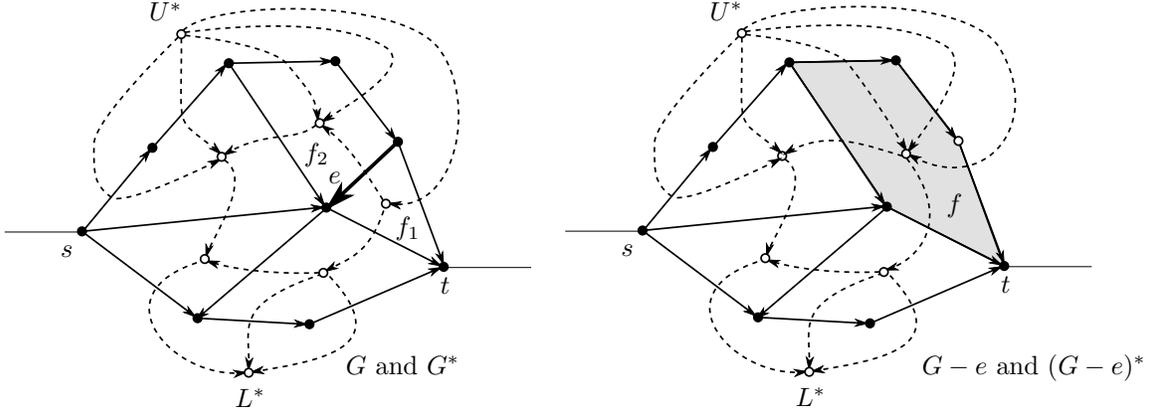
\end{center}
\caption{$st$-planar embedded graphs $G$ and $G-e$ (solid) and their dual graphs $G^{*}$ and $(G-{e})^{*}$  (dotted). Removing arc $e$, faces $f_{1}$ and $f_{2}$ in $G$ are merged, yielding face $f$ in $G-e$. The dual graph includes reverse arcs, with 0 length, for each dotted arc drawn in the picture.}\protect\label{fi:dual}
\end{figure}	

Given a planar embedded graph $G$, we say a set $S$ of arcs is \emph{contiguous} if the set of dual arcs $S^{*}$ defines a connected component in $G^{*}$. Note that the contiguity property depends on the embedding of the graph.

\bigskip
Our algorithms rely on the following result:
\begin{lemma}\label{le:general}
For  each arc $e$,
$$\flowvit(e) =
\max\left\{
0,
c(\mincut(G)) -\left(c(\mincut_e(G))  - c(e)\right)
\right\}
$$
\end{lemma}
\begin{proof}
By the Min-Cut Max-Flow theorem, $\flowvit(e) = c(\mincut(G)) - c(\mincut(G-e))$.
Obviously,
$c(\mincut(G-e)) \leq c(\mincut(G))$.

Let us first assume that $c(\mincut(G-e)) < c(\mincut(G))$. In this case, $\flowvit(e) = c(\mincut(G)) - c(\mincut(G-e)) > 0$, and $\mincut(G-e) \cup \{e\}$ is a minimum capacity cut among all cuts crossing arc $e$.
In fact, by contradiction, if a cut $C'$ crossing $e$ exists with capacity smaller than $c(\mincut(G-e)) + c(e)$, then $C' \setminus \{e\}$ would be a cut in $G - e$ with capacity smaller that $c(\mincut(G-e))$. 

Hence,
$\mincut_{e}(G) \setminus \{e\}$ is a minimum cut in $G-e$, and its capacity is $c(\mincut_{e}(G)) - c(e)$.
Otherwise, let  $c(\mincut(G-e)) = c(\mincut(G))$. In this case, by definition, $\flowvit(e) = 0$ and $c(\mincut_{e}(G)) - c(e) \geq c(\mincut(G))$, thus giving the thesis.
\end{proof}

\section{General undirected graphs}\label{se:undirected}
By definition, the vitality of each arc $e$ can be computed by solving a max-flow problem on $G-e$, therefore, we can compute the vitality of all arcs by $m$ calls to a max-flow routine.

In the following, we show that, in the case of general undirected graphs, it is possible to compute the vitality of all arcs by only $2(n-1)$ calls to a max-flow routine.
Lemma~\ref{le:general} shows that, in order to compute $\flowvit(e)$ for any given arc $e=\edge{x}{y}$ in a general undirected graph, it is sufficient to compute $c(\mincut_{e}(G))$. Let $(C, \overline{C})$ be any $st$-cut that crosses $e$, then either:
\begin{enumerate}[(1)]
\item\label{item:cut1} $\{x,s\} \subseteq C$ and $\{y,t\} \subseteq \overline{C}$ or
\item\label{item:cut2} $\{y,s\} \subseteq C$ and $\{x,t\} \subseteq \overline{C}$.
\end{enumerate}

Therefore, we can find the minimum capacity cut by comparing the best cut of type (\ref{item:cut1}) and the best cut of type (\ref{item:cut2}). A minimum capacity cut of type (\ref{item:cut1}) can be found by applying a standard min-cut algorithm to a graph $G'$ obtained from $G$ by adding two arcs $\edge{x}{s}$ and $\edge{y}{t}$ with very high capacities (e.g., greater than the sum of all the capacities in the graph). Obviously, a minimum $st$-cut in $G'$  cannot separate $x$ from $s$, nor can it separate $y$ from $t$, thus it necessarily crosses arc $e$.  Analogously, minimum capacity cuts of type  (\ref{item:cut2}) are found by adding high capacity arcs $(x,t)$ and $(y,s)$.

Computing cuts of type (\ref{item:cut1}) and of type (\ref{item:cut2}) for each arc in $G$ can be done by solving $2m$ minimum cut problems, but the number of minimum cuts to be computed is actually much smaller.
Gomory and Hu~\cite{Gomory} showed that in any undirected graph a set $\mathcal{C}$ of at most $n-1$ cuts exists so that for each pair of nodes $x,y$ a minimum $xy$-cut can be found in $\mathcal{C}$. This means that the $\binom{n}{2}$ pairs of nodes in $G$  can be  separated by using only $n-1$ different minimum cuts. Moreover, these cuts can be implicitly represented by a \emph{cut tree}:
\begin{definition}[\cite{Gomory}]
A \emph{cut tree} $T=(N,A_{T},w)$ of a weighted undirected graph $G=(N,A,c)$ is a  tree with real weighted arcs that represents minimum capacity cuts for all pairs of nodes in $N$. More precisely, for any two nodes $x,y \in N$, let $e$ be a minimum weight arc in the unique path joining $x$ and $y$ in $T$: then $T$ is a cut tree of $G$ if and only if
\begin{enumerate}[(i)]
\item\label{cutweight} $w(e)$ equals the capacity of a minimum $xy$-cut in $G$, and
\item\label{cut} $e$ splits $T$ into two connected components with node sets $X$ and $Y$ so that $(X,Y)$ is a minimum capacity $xy$-cut in $G$.
\end{enumerate}
\end{definition}

A first algorithm for computing a cut tree has been proposed in~\cite{Gomory}, and a simpler approach is shown in~\cite{Gusfield}.

A \emph{flow tree} differs from a cut tree in the fact that property~(\ref{cut}) is not required. Hence, a flow tree only represents the \emph{capacities} of the minimum cuts, not the cuts themselves, for all pairs of nodes. A very simple algorithm for computing a flow tree is given in~\cite{Gusfield}. All the algorithms in~\cite{Gomory,Gusfield} require $n-1$ maximum flow computations.

The concept of flow tree has been generalized by Cheng and Hu in~\cite{ancestor}. In their more general setting, an arbitrary real function $f$ is defined on the set of cuts (i.e., node bipartitions) and, given any two nodes $x,y$, an $xy$-cut that minimizes $f$ has to be computed.
\begin{definition}[\cite{ancestor}]
Given an undirected graph $G=(N,A)$ and a real function $f$ defined on the set of all bipartitions of $N$, an \emph{ancestor tree} $T_{f}$ is a binary tree with leaves $N$ such that each internal node $v$ of $T_{f}$ represents a minimum cut (w.r.t.\ $f$) separating each leaf in the left subtree of $v$ from each leaf in the right subtree of $v$.
\end{definition}
 
The minimum (w.r.t.\ $f$) cut separating two nodes $x$ and $y$ in $G$ can be found by looking at the lowest common ancestor of $x$ and $y$ in $T_{f}$.
For example, we can define $f$ so that balanced bipartitions are preferred, or impose any other arbitrary constraint and/or cost function to cuts. As a special case, $f$ can be the sum of the capacities of arcs crossing the cut, as in the classical max-flow problem.
Cheng and Hu showed that for any undirected graph $G$ and any cost function $f$ it is always possible to compute an ancestor tree $T_{f}$.
Assuming that, given a pair $x,y$, a routine is available for computing a minimum $xy$-cut in $G$ according to the cost function $f$, building an ancestor tree requires $n-1$ calls to that routine---plus overall $O(n^{2})$ worst-case time for restructuring operations. 

Note that, while an ancestor tree exists for any cost function $f$, it is not always possible to define a cut tree according to $f$. For example, let us define $f$ as in the classical min cut problem, with the exception that partitions in which one side contains only one node have cost $+\infty$: a cut tree should have at least one leaf $v$, and the cut defined by the arc incident on $v$ in the cut tree would define a partition in which one side contains only node $v$.

The structure of the solution space of generalised cut problems is studied in~\cite{Hartvigsen01,Hassin88,Hassin90}.

\vspace{5mm}
We are now ready to describe our algorithm for computing arc vitalities for general undirected graphs.
We first compute $c(\mincut(G))$, and then we build an ancestor tree $T_{st}$ according to cost function $f_{st}$ defined as follows:
$$f_{st}(C,\overline{C}) =
\left\{
\begin{array}{ll}
+\infty & \mbox{if}\ \{s,t\} \subseteq C\ \mbox{or}\  \{s,t\} \subseteq \overline{C}\\[2mm]
c(C,\overline{C}) & \mbox{otherwise} 
\end{array}
\right.
$$
For each arc $e=\edge{x}{y}$ in $G$, we find on $T_{st}$  the capacity of a minimum $xy$-cut that also separates $s$ 
from $t$.
Minimizing the cost function $f_{st}$  gives $\mincut_{e}(G)$, for each $e$, and, by Lemma~\ref{le:general}, allows us to compute $\flowvit(e)$ in constant time.

\begin{theorem}
Given an undirected weighted graph $G=(N,A,c)$ and two nodes $s,t$, we can compute $\flowvit(e)$, for all $e \in A$, in $O(n \cdot \mbox{MF}(n,m) + m \cdot n)$ worst-case time, where $\mbox{MF}(n,m)$ is the time needed to compute a maximum flow.
\end{theorem}
\begin{proof}
Building the ancestor tree $T_{st}$ requires $n-1$ calls to a routine that, given two nodes $x,y$, computes a minimum cut that separates both $x$ from $y$ and $s$ from $t$. As described in the beginning of this section, such a cut can be found by solving two standard max-flow instances, namely, on a graph $G'$ obtained from $G$ by adding two arcs $\edge{x}{s}$ and $\edge{y}{t}$ with very high capacities
and on a graph $G''$ obtained from $G$ by adding two arcs
$\edge{x}{t}$ and $\edge{y}{s}$ with very high capacities.

For each arc $e=\edge{x}{y}$, the value $\flowvit(e)$ can be computed by Lemma~\ref{le:general}, where $c(\mincut_{e}(G))$ is found on $T_{st}$ by searching for the lowest common ancestor of $x,y$. This trivially\footnote{Lowest common ancestors could be found more efficiently, but in our case this is not the dominant asymptotic cost} requires $O(n)$ for each arc, leading to an overall $O(m \cdot n)$ additional worst-case time.
\end{proof}

By applying the currently fastest algorithms for max-flow in general graphs, i.e., King, Rao and Tarjan's algorithm~\cite{KingRaoTarjan} for $m = \Omega(n^{1+\varepsilon})$ with $\varepsilon > 0$, and Orlin's algorithm~\cite{Orlin} for sparse graphs, both requiring worst-case time $O(m \cdot n)$, we can state the following result.
\begin{corollary}
For any undirected graph $G=(N,A,c)$, we can compute $\flowvit(e)$, for all $e \in A$, in $O(n^{2} \cdot m)$ worst-case time.
\end{corollary}
 
In order to appreciate the efficiency of the above algorithm for computing the vitality of all arcs, observe that almost all arcs in a graph might have a non-trivial vitality, i.e., there can be $\Omega(n^{2})$ arcs such that $0 < \flowvit(e) < c(e)$. This is shown by graph $B$ in Figure~\ref{fi:example}, with $n=2k + 2$ nodes $\{s, x_{1}, x_{2}, \ldots, x_{k}, y_{1}, y_{2}, \ldots, y_{k}, t\}$. Nodes $x_{i}$'s and $y_{i}$'s are the two sides of a complete bipartite graph, and there are $k$ arcs joining $s$ to each $x_{i}$ and $k$ arcs joining $t$ to each $y_{i}$. Each arc $(x_{i}, y_{j})$, for $1 \leq i,j \leq k$, has a distinct capacity $c((x_{i}, y_{j})) = 1 + \varepsilon_{i,j}$,  with $0 < \varepsilon_{i,j} < \frac{1}{k}$, while all arcs $(s, x_{i})$ and $(t, y_{i})$, for $1 \leq i \leq k$,  have capacity $k$. We show here that $0 < \flowvit((x_{i}, y_{j})) < 1$ for all $1 \leq i,j \leq k$.

A maximum flow assignment in $B$ is given by $f(s, x_{i})=k$, for $1 \leq i \leq k$, and $f(y_{j}, t)=k$, for $1 \leq j \leq k$, and $f(x_{i}, y_{j})=1$, for $1 \leq i,j \leq k$. Thus:
$$\maxflow(B) = k^2$$

Let us evaluate $\flowvit((x_{i}, y_{j}))$. Assuming w.l.o.g.\ that
$$\sum_{\substack{1\leq q \leq k\\ \ q \not=j}} (\varepsilon_{i,q})
\leq \sum_{\substack{1\leq p \leq k\\ \ p \not=i}} (\varepsilon_{p,j})$$
 a minimum $st$-cut in $B - (x_{i},y_{j})$ is given by bipartition $\left( \{s, x_{i}\}, V(B) \setminus  \{s, x_{i}\} \right)$, and its crossing arcs are $(s, x_{p})$, for $p \not= i$, together with arcs $(x_{i},y_{q})$, for $q \not= j$, giving 
$$\maxflow(B - (x_{i}, y_{j}))
= k(k-1) + \sum_{\substack{1\leq q \leq k\\ \ q \not=j}} (1 + \varepsilon_{i,q})
= k^{2} - 1 + \sum_{\substack{1\leq q \leq k\\ \ q \not=j}} \varepsilon_{i,q}
$$

Hence, since $0 < \varepsilon_{i,q} < \frac{1}{k}$,
$$\flowvit((x_{i}, y_{j})) = 1 - \sum_{\substack{1\leq q \leq k\\ \ q \not=j}} \varepsilon_{i,q} \in (0,1)$$

Values $\varepsilon_{i,j}$ can easily be chosen so that all arcs have distinct vitalities.

\begin{figure}[t]
\begin{center}
\def\svgwidth{10.0cm}
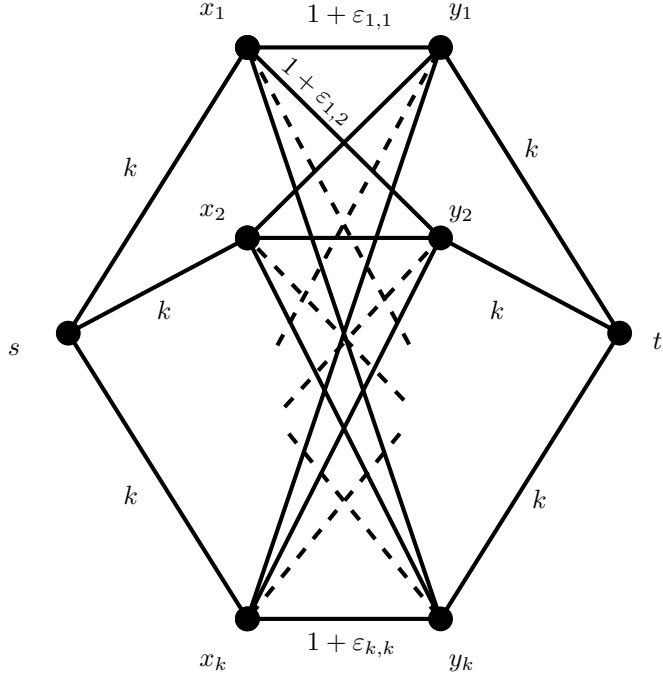
\end{center}
\caption{Graph $B$. All the $k^{2}$ arcs in the complete bipartite graph have distinct non-trivial vitality.}\protect\label{fi:example}
\end{figure}	

\section{$st$-planar graphs}\label{se:stplanar}
Algorithms for finding the maximum flow in a planar graph vary according to whether the input graph is directed or undirected, and whether it is $st$-planar, i.e., it can be drawn on a plane with $s$ and $t$ on the same face.

In the case of  undirected planar (but not necessarily $st$-planar) graphs,
Reif~\cite{Reif} proposed a divide and conquer approach for computing a minimum $st$-cut in $O(n\log^2n)$ worst-case time. By plugging in the SSSP tree algorithm for planar graphs by Henzinger et al.~\cite{Henzinger}, this bound can be improved to $O(n\log n)$. The best currently known approach for computing a minimum $st$-cut is due to Italiano et al.~\cite{Italiano}, and it achieves $O(n\log \log n)$ time by a two phase approach, that exploits the algorithm by Hassin and Johnson~\cite{Hassin85}.
For directed planar graphs, Borradaile and Klein~\cite{Borradaile} presented an $O(n\log n)$ time algorithm based on a repeated search of left-most circulations.

The first algorithm proposed for directed  $st$-planar graphs is due to Ford and Fulkerson~\cite{Fulkerson} and consists in repeatedly saturating the uppermost path of a planar embedding of the graph, each time deleting saturated arcs. 
Itai and Shiloach~\cite{Itai} proposed an $O(n\log n)$ time implementation of this procedure, by using a priority queue for finding the saturating arc of each uppermost path.
Later, Hassin~\cite{Hassin} proved that, if $G$ is $st$-planar and $G^*$ is the dual of an $st$-planar embedding of $G$,
a minimum $st$-cut in $G$ corresponds to a shortest path in $G^*$ and the maximum flow can be computed in linear time starting from a single source shortest path (SSSP) tree. In fact, Hassin~\cite{Hassin}  showed that the SSSP tree also defines a maximum flow assignment. It suffices to assign to each arc $e$ a flow value equal to the difference between the distances from the lowermost face to the two endpoints of $e^{*}$.
Using the algorithm by Henzinger et al.~\cite{Henzinger} for the SSSP tree problem in planar graphs, the minimum $st$-cut in $G$ and the corresponding maximum flow can be found in $O(n)$ time.  

We describe in Section~\ref{sse:starcs} how the ideas in~\cite{Itai} can be applied to compute the vitality of all arcs in an $st$-planar graph.
The same approach
 is then generalised in Section~\ref{use:stnodes} to compute the vitality of nodes and of more general contiguous sets of arcs.
 
\subsection{Vitality of arcs}\label{sse:starcs}

Fixing an $st$-planar embedding of $G$, we can exploit the strong correspondence for $st$-planar graphs between flows in $G$ and distances from the upper node  $U^{*}$ to the lower node $L^{*}$ in $G^{*}$. The definition of vitality  gives, for  each arc $e$:
\begin{equation*}
\flowvit(e) = \dist_{G^{*}}(U^{*}, L^{*}) - \dist_{G_{e}^{*}}(U^{*}, L^{*})
\end{equation*}
Obviously, $\dist_{G_{e}^{*}}(U^{*}, L^{*}) \leq \dist_{G^{*}}(U^{*}, L^{*})$. 
An analog version of Lemma~\ref{le:general}  allows us to state the same equality using only distances in $G^{*}$.
\begin{lemma}\label{le:faces}
Given an $st$-planar embedded directed or undirected graph $G$ and an arc $e \in G$, we have:
$$\flowvit(e) = \max
            \left\{
                 \begin{array}{l}
                 0, \\
                 \dist_{G^{*}}(U^{*},L^{*}) - 
                  \left(
                                                   \dist_{G^{*}}(U^{*},e^{*})
                                          +
                                                   \dist_{G^{*}}(e^{*},L^{*})
                     \right)
                    \end{array}
              \right\}
$$
\end{lemma}
\begin{proof}
Let us first consider  directed graphs.
Graph $G^{*}_{e}$ derives from $G^{*}$ after setting $\lungh(e^{*}) = \lungh((f_{1}^{*},f_{2}^{*}))=0$, where $f_{1}$ and $f_{2}$ are the faces in $G$ respectively to the left and to the right of $e$. This also corresponds, in the primal graph, to merging faces $f_{1}$ and $f_{2}$ into a single face $f$ (see Figure~\ref{fi:dual}).
Obviously, $\dist_{G_{e}^{*}}(U^{*}, L^{*}) \leq \dist_{G^{*}}(U^{*}, L^{*})$. 

A shortest path $\pi$ from $U^{*}$ to $L^{*}$ in $G^{*}_{e}$ is either a path that does not contain $e^{*}$, or is the concatenation of a shortest path $\pi_{1}$ from $U^{*}$ to $f_{1}^{*}$, arc $e^{*}$, and  a shortest path $\pi_{2}$ from $f_{2}^{*}$ to $L^{*}$. If $\pi$ does not contain $e^{*}$, then $\pi$  is also a shortest path from $U^{*}$ to $L^{*}$ in $G^{*}$. If $\pi$ contains $e^{*}$, then $\pi_{1}$ and $\pi_{2}$ do not intersect each other: in fact, if $\pi_{1}$ and $\pi_{2}$ intersect, a path shorter than $\pi$ exists, as shown in Figure~\ref{fi:pathIntersecanti}.
\begin{figure}[t]
\begin{center}
\def\svgwidth{8cm}
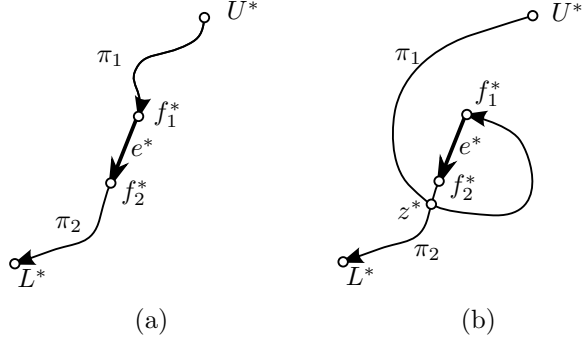
\end{center}
\caption{(a) A shortest path from $U^{*}$ to $L^{*}$ passing through $e^{*}$ is the concatenation of $\pi_{1}$, $e^{*}$, and $\pi_{2}$.
(b) If $\pi_{1}$ intersects $\pi_{2}$ in $z^{*}$, then a shorter path from $U^{*}$ to $L^{*}$ through $z^{*}$  exists.}\protect\label{fi:pathIntersecanti}
\end{figure}	
Thus,
\begin{equation}\label{eq:varicasidirected}
\dist_{G^{*}_{e}}(U^{*},L^{*})
=
\min\left\{
	\begin{array}{l}
	\dist_{G^{*}}(U^{*},L^{*})\\

	\dist_{G^{*}}(U^{*},f_{1}^{*})  +  \dist_{G^{*}}(f_{2}^{*},L^{*})
	\end{array}
\right\}
\end{equation}

In the first case $\flowvit(e) = 0$, while in the second case
$\flowvit(e) = \dist_{G^{*}}(U^{*},L^{*}) - 
                  \left(
                                                   \dist_{G^{*}}(U^{*},e^{*})
                                          +
                                                   \dist_{G^{*}}(e^{*},L^{*})
                     \right)
$.

In the undirected case, a shortest path from $U^{*}$ to $L^{*}$ containing $e^{*}$ could visit either $f_{1}^{*}$ or $f_{2}^{*}$ first. Thus, instead of equality~(\ref{eq:varicasidirected}), we have
\begin{equation*}
\dist_{G^{*}_{e}}(U^{*},L^{*})
=
\min\left\{
	\begin{array}{l}
	\dist_{G^{*}}(U^{*},L^{*})\\

	\dist_{G^{*}}(U^{*},f_{1}^{*})  +  \dist_{G^{*}}(f_{2}^{*},L^{*})\\
	\dist_{G^{*}}(U^{*},f_{2}^{*})  +  \dist_{G^{*}}(f_{1}^{*},L^{*})
	\end{array}
\right\}
\end{equation*}
Hence, also for undirected graphs, either $\flowvit(e) = 0$, or
$\flowvit(e) = \dist_{G^{*}}(U^{*},L^{*}) - 
                  \left(
                                                   \dist_{G^{*}}(U^{*},e^{*})
                                          +
                                                   \dist_{G^{*}}(e^{*},L^{*})
                     \right)
$.

\end{proof}

\begin{theorem}
Given a directed or undirected $st$-planar weighted graph $(N, A, c)$, we can compute $\flowvit(e)$, for all $e \in A$, in $O(n)$ worst-case time.
\end{theorem}
\begin{proof}
Lemma~\ref{le:faces} only uses graph $G^{*}$, thus it is possible to compute the vitality of arc $e$ without explicitly computing $G^{*}_{e}$, for each arc $e$.
It suffices to store, for each node $f^{*}$ in $G^{*}$, the pair of distances $ \dist_{G^{*}}(U^{*},f^{*})$ and 
$\dist_{G^{*}}(f^{*},L^{*})$. These can be computed by means of two single-source shortest path trees,
the first in $G^{*}$ from $U^{*}$ to each other node and the second in the reversal of $G^{*}$ from $L^{*}$ to each other node. Shortest path trees can be found in planar graphs in $O(n)$ worst-case time using the technique in~\cite{Henzinger}.
\end{proof}

\subsection{Vitality of nodes and contiguous arc sets}\label{use:stnodes}
Let $F$ be a set of contiguous arcs in an $st$-planar embedding of a directed or undirected $st$-planar graph $G$. We recall that the set of the dual arcs $F^{*}$ defines  a connected subgraph in $G^{*}$. In particular, in case $F$ is the set of all arcs incident on the same node, then $F$ is contiguous in any $st$-planar embedding of $G$.

By definition, $\flowvit(F) = \maxflow(G) - \maxflow(G-F)$ and, thanks to the Min-Cut Max-Flow theorem and the result in~\cite{Hassin},  $\flowvit(F) = \dist_{G^{*}}(U^{*},L^{*}) - \dist_{(G-F)^{*}}(U^{*},L^{*})$.
The dual graph $(G-F)^{*}$ is obtained from $G^{*}$ by setting to 0 the length of all arcs in $F$. Let $K^{*}$ be the set of endpoints of all arcs in $F^{*}$.
If $G$ is undirected, then $K^{*}$ is connected. If $G$ is directed,
since for each arc $(x^{*}, y^{*})$ in $F^{*}$ a reverse arc $(y^{*}, x^{*})$ exists with $w(y^{*}, x^{*}) = 0$, then  $K^{*}$ is strongly connected.
In any case, each node in $K^{*}$ has the same distance from/to any given node in $G^{*}$, including $U^{*}$ and  $L^{*}$. Thus, $\dist_{(G-F)^{*}}(U^{*},x^{*}) = \min_{v^{*} \in K^{*}}(\dist_{G}(U^{*},v^{*}))$, for each $x^{*} \in K^{*}$.

It follows that $\flowvit(F)$ can be obtained again using only distances in $G^{*}$.
In fact,
\begin{equation}\label{eq:mindist}
\dist_{(G-F)^{*}}(U^{*},L^{*}) = 
\min \left\{
	\begin{array}{l}
	\dist_{G^{*}}(U^{*},L^{*})\\
     \dist_{G^{*}}(U^{*},K^{*}) + \dist_{G^{*}}(K^{*},L^{*})
     \end{array}
\right\}
\end{equation}
since the portion of any path in $(G-F)^{*}$ inside $K^{*}$ has zero length.
By definition, $ \dist_{G^{*}}(U^{*},K^{*})$ and $\dist_{G^{*}}(K^{*},L^{*})$ can be computed in $O(|K|)$ time as $\min_{a^{*} \in K^{*}}\{\dist_{G^{*}}(U^{*},a^{*})\}$ and\\
$\min_{b^{*} \in K^{*}}\{\dist_{G^{*}}(b^{*},L^{*})\}$, respectively.

The above argument yields the following theorem:
\begin{theorem}\label{th:contiguous}
Given an $st$-planar embedding of a directed or undirected $st$-planar graph $G$, it is possible to preprocess $G$ in $O(n)$ worst-case time so that, for any set of contiguous arcs $F$, $\flowvit(F)$ can be answered in $O(|F|)$ worst-case time.
\end{theorem}

The arguments leading to Theorem~\ref{th:contiguous} can be specialized in order to compute the vitality of nodes.
Deleting a node $x$ corresponds to deleting the set of all arcs incident on $x$, and this set of arcs is contiguous in any planar embedding of $G$. 
Thanks to equation~(\ref{eq:mindist}), $\dist_{(G-x)^{*}}(U^{*},L^{*})$ can be derived by computing $\dist_{G^{*}}(U^{*},N^{*}_x)$, where $N_x$ is the set of faces adjacent to $x$ and $N^*_x$ is the corresponding set of nodes in $G^*$. This is the minimum among distances (in the dual graph $G^{*}$) from $U^{*}$ to all nodes corresponding to faces surrounding $x$, and similarly from all these nodes to $L^{*}$.
Thus, after the two SSSP trees from $U^{*}$ and to $L^{*}$ have been computed in $O(n)$ time as in~\cite{Henzinger}, we can associate to each node $x$ the values of $\dist_G^{*}(U^{*},N^{*}_x)$ and $\dist_{G^{*}}(N^{*}_x, L^{*})$, still in overall $O(n)$ worst-case time. The above arguments lead to the following theorem.
\begin{theorem}\label{th:nodes}
Given an $st$-planar embedding of a directed or undirected $st$-planar graph $G$, it is possible to compute   $\flowvit(x)$, for all nodes $x$, in $O(n)$ worst-case time.
\end{theorem}


\section{Conclusions and further work}\label{se:conclusions}
In this paper we have shown how to solve the problem of computing the vitality of all arcs and all nodes  with respect to the maximum flow for general undirected graphs and for $st$-planar graphs. 
For general undirected graphs, we compute the vitality of all arcs applying $O(n)$ time a max-flow algorithm.
In $st$-planar graphs (directed and undirected) our algorithm runs in optimal $O(n)$ worst-case time.
Some points are left open.

First, we point out that the technique we use for general undirected graphs, based on ancestor trees, cannot be directly applied to the directed case, since in~\cite{ancestor} it is implicitly assumed that the objective function $f$ defining minimum cuts is symmetric, i.e., $f(A,B) = f(B,A)$. Therefore, ancestor trees do not necessarily exist in the directed case and, for the time being, we do not know of any non-trivial solution  for computing the vitalities of all arcs---i.e., better than the $O(m \cdot \mbox{MF}(m,n))$ worst-case time obtained by applying a standard max-flow algorithm after deleting each single arc.

In the second place, the general undirected and directed planar cases are also open. 
In fact, the divide and conquer technique by Reif~\cite{Reif} for finding min-cuts in undirected planar graphs cannot be directly applied to compute the vitality of arcs.
Moreover, as it is the case for general graphs, also for planar graphs minimum cuts have a richer structure in the directed case than in the undirected one.  As shown in~\cite{Lacki},
in the directed case the set of minimum $xy$-cuts defined by all pairs of nodes $x,y$ may contain $\Theta(n^{2})$ different minimum cuts, while there are always at most $n-1$ different $xy$-cuts in the undirected case.

\section*{Acknowledgements}
We wish to thank Nicola Apollonio for many deep discussions on this topic. We are grateful to the anonymous reviewers for their invaluable comments, and for pointing out a flaw in a previous version of the paper.
\bibliographystyle{plain}
\bibliography{vital}

\end{document}